\documentclass{article}
	\usepackage{amsmath,amssymb,amsthm}
\newtheorem{lemma}{Lemma}
\newtheorem{prop}{Proposition}
\newtheorem{theorem}{Theorem}
\newtheorem{coro}{Corollary}

\theoremstyle{remark}
\newtheorem{ex}{Example}
\newtheorem{rem}{Remark}

\def\states{\mathfrak S}
\def\Tr{\mathrm{Tr}\,}
\def\Ce{\mathcal C}
\def\Fe{\mathcal F}
\def\Ha{\mathcal H}
\def\Me{\mathcal M}
\def\Te{\mathcal T}
\def\Ka{\mathcal K}

\def\ptr{\mathrm{Tr}}
\def\<{\langle}
\def\>{\rangle}
\def\kk{\rangle\!\rangle}
\def\bb{\langle\!\langle}

\def\supp{\mathrm{supp}\,}
\def\Le{\mathcal L}
\def\Ae{\mathcal A}
\def\prob{\mathsf P}

\begin{document}

\title{On the convex structure of process POVMs}
\author{Anna Jen\v cov\'a\\
Mathematical Institute, Slovak Academy of Sciences\\
\v Stef\'anikova 49, 814 73 Bratislava, Slovakia\\
\texttt{jenca@mat.savba.sk}}
\date{}
\maketitle
\begin{abstract}Measurements on quantum channels are described by so-called process operator valued measures, or process POVMs. We study implementing schemes of extremal process POVMs. As it turns out, the corresponding measurement must satisfy certain extremality property, which is stronger that the usual extremality given by the convex structure. This property motivates the introduction and investigation of the $\Ae$-convex structure of POVMs, which generalizes both the usual convex  and C*-convex structure. We show that extremal points and faces of the set of process POVMs are closely related to $\Ae$-extremal points and $\Ae$-faces of POVMs, for a certain subalgebra $\Ae$. We give a characterization of $\Ae$-extremal and $\Ae$-pure POVMs in the Appendix. 
\end{abstract}

\section{Introduction and basic definitions}


Process positive operator valued measures, or process POVMs, were introduced by Ziman \cite{ziman2008PPOVM}, as a mathematical tool for description of measurements on quantum channels. Similar to the usual POVMs, which represent quantum observables, process POVMs are sequences of positive operators, or more generally $\sigma$-additive measures with values in the set of positive operators, but satisfying a different normalization condition. Independently, the same concept, called quantum 1-testers, was studied by \cite{cdp2008memory, cdp2009framework}, and also in \cite{guwa2007games}, as measuring quantum co-strategies. 

Similarly to other quantum devices, the set of all process POVMs is convex.  In many cases, the convex  structure determines the performance of the corresponding measurements, for example optimal measurements with respect to convex figures of merit are given by extremal process POVMs. On the other hand, the set of process POVMs can also be the subject of statistical inference and the convex structure plays a decisive role in discrimination tasks, see \cite{jencova2014base,pjsz2015exploring}.

Physically, any measurement on channels can be realized by applying the channel on a part of an input state $\rho$ and consequently measuring the outcome by a usual POVM $M$. The aim of the present paper is to describe the extreme points and faces of the set  of process POVMs in terms of these implementing schemes.  It is easy to see that there is a lot of such schemes for the same process POVM, but under certain minimality conditions the input state $\rho$ and measurement $M$ are unique up to a  unitary conjugation. It was shown in \cite{bdps2011extremal,jencova2012extremalityJMP} that a process POVM is extremal if and only if the POVM $M$ in such a  minimal representation is also an extremal process POVM (note that any POVM is a multiple of a process POVM, describing a channel measurement with maximally entangled input state). In the present work, we obtain  a characterization within the  set of POVMs, more precisely, in terms of its C*-convex structure.

The notion of C*-convexity of subsets of operators was introduced and studied in \cite{lopa1981some,hmp1981Cstar,famo1993Cstar, morenz1994thestructure}. Roughly speaking, instead of numbers in the interval $(0,1)$, the coefficients of a C*-convex combination are operators forming a resolution of the identity. The C*-extreme points of the set of POVMs, or more generally of the set of unital completely positive maps from a C*-algebra into the algebra $B(\Ha)$ of bounded  operators  on a finite dimensional Hilbert space, were studied in \cite{famo1997Cstar, fazh1998thestructure}. In particular, it was proved that a POVM is C*-extremal if and only if it is projection valued. We propose a natural extension of this notion, containing both C*-convexity and usual convexity, such that the coefficients of convex combinations are restricted to a given subalgebra $\Ae\subseteq B(\Ha)$. We show that extremal elements and faces of the set of process POVMs correspond to $\Ae$-extremal elements and $\Ae$-faces of the representing POVMs, for some subalgebra $\Ae$. We also describe the face generated by a process POVM, implemented by a scheme with $\Ae$-extremal POVM. In particular, we show that such a process POVM is not necessarily extremal, but any element in its convex decomposition is obtained by a choice of the input state, while keeping the same measurement.

The outline of the paper is as follows. In the rest of this section, we introduce some basic notations and definitions. In Section \ref{sec:aconv} we introduce the notion of $\Ae$-convexity and state some of its properties. All the results needed in this paper are proved similarly as in the C*-convex case, we give some of the proofs in the Appendix. A more general characterization of $\Ae$-extremal generalized states on a C*-algebra $\mathcal B$ will be given elsewhere. Process POVMs and their representing triples are discussed in Section \ref{sec:ppovm}, Section \ref{sec:extr} contains the main results.

Let $\Ha$ be a Hilbert space, with $d_\Ha:=\dim(\Ha)<\infty$. We denote the set of (bounded) linear operators on $\Ha$ by $B(\Ha)$, the set of  positive operators by $B(\Ha)^+$ and the set of density operators, that is positive operators with unit trace, by $\states(\Ha)$. 

Measurements on the system are represented by positive operator valued measures (POVMs).
In general, these are defined on a $\sigma$-algebra of measurable subsets of the set $X$ of outcomes, but we will only deal with $X=\{1,\dots,n\}$. In this case, a POVM is a collection of positive operators $M_1,\dots,M_n$ on $\Ha$, such that $\sum_iM_i=I$. If the system is in some state $\rho\in \states(\Ha)$, the probability of obtaining the outcome $i\in X$ is given by $\Tr \rho M_i$. In this way, POVMs correspond precisely to affine maps from the state space to the probability simplex  $\prob_n$.
We will denote the set of $n$-outcome POVMs by $\Me(\Ha,n)$.

\begin{rem}\label{rem:ext} Let $\Ka\subseteq \Ha$ be a subspace and let $P:\Ha\to \Ka$ be the corresponding projection. For any $M\in \Me(\Ka,n)$, we 
put $\widetilde M_i=M_i+\tfrac1n (I-P)$, then $\widetilde M\in \Me(\Ha,n)$. In this way we may, and will, identify 
$\Me(\Ka,n)$ with a subset of $\Me(\Ha,n)$.

\end{rem}

 Let $\Le(\Ha,\Ka)$ denote the set of linear maps $B(\Ha)\to B(\Ka)$  and 
 let $\Ce(\Ha,\Ka)$ be the set of channels, that is, completely positive trace preserving maps. The Choi isomorphism 
\[
C: \phi\mapsto (\phi\otimes id_\Ha)(\psi_\Ha)\in B(\Ka\otimes \Ha)
\]
maps $\Le(\Ha,\Ka)$ onto $B(\Ka\otimes \Ha)$ and the set of completely positive maps  onto $B(\Ka\otimes \Ha)^+$. Here
\[
\psi_\Ha=\sum_{i,j} |i\>\<j|\otimes |i\>\<j|
\]
for some fixed ONB $\{|i\> \}$ in $\Ha$ and $id_\Ha$ is the identity map. 

Let $|\xi\>\in \Ha\otimes \Ha_0$ be a unit vector. Then there are orthonormal 
bases $|e_1\>,\dots,|e_{d_\Ha}\>$ and $|f_1\>,\dots,|f_{d_{\Ha_0}}\>$ of $\Ha$ and $\Ha_0$, such that 
\[
|\xi\>=\sum_{i=1}^k \alpha_i |e_i\>\otimes |f_i\>,
\]
with $\alpha_i>0$, this is the Schmidt decomposition of $|\xi\>$.
The number $k\le d_\Ha\wedge d_{\Ha_0}$ of nonzero 
coefficients is called the {Schmidt rank} of $|\xi\>$ and is denoted by $SR(\xi)$. Clearly, $SR(\xi)$ is the rank of 
$\ptr_\Ha|\xi\>\<\xi|$ or $\ptr_{\Ha_0}|\xi\>\<\xi|$. If $\rho$ is a mixed state in $\states(\Ha\otimes \Ha_0)$, then its {Schmidt number}  $SN(\rho)$ is the smallest $k$ such that 
$\rho$ can be expressed as a convex combination of pure states with Schmidt rank at most  $k$.
Clearly, $\rho$ is separable if and only if $SN(\rho)=1$.

Let $\xi_1,\dots,\xi_{d_\Ha}\in \Ha_0$ be such that $|\xi\>=\sum_i |i\>\otimes |\xi_i\>$. Let $T:\Ha\to \Ha_0$ be the linear operator defined by $T|i\>=|\xi_i\>$, then $\Tr T^*T=1$, moreover, 
\[
|\xi\>=|T\kk:=\sum_i |i\>\otimes T|i\>
\]
and the map $T\mapsto |T\kk$ defines a linear isomorphism of the set of linear operators $\Ha\to \Ha_0$ 
onto $\Ha\otimes \Ha_0$. We have $SR(|T\kk)=\mathrm{rank}(T)$.

\section{The $\Ae$-convex structure of POVMs}\label{sec:aconv}

In this paragraph, we introduce the notion of $\Ae$-convexity, $\Ae$-extreme 
points and $\Ae$-faces for POVMs on $B(\Ha)$, where $\Ae$ is a C*-subalgebra in $B(\Ha)$. The results used in this paper are similar to the C*-convex case, see Example \ref{ex:cstar} below. We postpone most of the proofs to the Appendix.

Let $\Ae\subseteq B(H)$ be a subalgebra and let  $M,N\in \Me(\Ha,n)$. We say that $M$ and $N$ are {$\Ae$-equivalent}, in notation $M\sim_\Ae N$, if there is a unitary $U\in \Ae$ such that $U^*M_iU=N_i$, for all $i$. Let $M, N^1,\dots,N^k\in B(H)$, then   $M$ is an {$\Ae$-convex combination} of $N^1,\dots,N^k$ if there are some elements $X_1,\dots,X_k\in \Ae$, $\sum_j X_j^*X_j=I$, such that $M=\sum_j X_j^*N^jX_j$, that is,
\begin{equation}\label{eq:a_convex}
M_i=\sum_j  X_j^*N^j_iX_j, \qquad i=1,\dots,n. 
\end{equation}
 An  $\Ae$-convex combination is called {proper} if $X_i$ is an invertible element in $\Ae$, for all $i$.
It is clear that the set $\Me(\Ha,n)$ is  {$\Ae$-convex}, in the sense that it contains all $\Ae$-convex combinations of its elements. A POVM $M\in \Me(\Ha,n)$  is  {$\Ae$-extremal} if whenever $M$ is a proper $\Ae$-convex combination of some POVMs $N^1,\dots, N^k$, then $N^j\sim_\Ae M$.

$\Ae$-convexity is a natural extension of the following two important cases.

\begin{ex}\label{ex:conv} Let $\Ae=\mathbb CI$, then $\Ae$-convexity coincides with the usual notion of convexity. In the general context of C*-algebras, the extremal elements of the set of POVMs, regarded as (completely) positive maps from a commutative C*-algebra into $B(\Ha)$, were characterized in \cite{arveson1969I}, see also 
\cite{stormer1974positive,partha1999extremal, dlp2005classical} for different formulations of the extremality condition. In our setting, the condition can be stated as follows: let $P_i$ be the projection onto the support of $M_i$, $i=1,\dots,n$. Then $M$ is extremal if and only if the subspaces $P_i\Ha$ are weakly independent, that is,  $D_i\in B(P_i\Ha)$ and $\sum_i D_i=0$ implies $D_i=0$ for all $i$ (compare this condition with Lemmas \ref{lemma:char_extreme} and \ref{lemma:char_pure} below). One can prove exactly the same way as for the C*-extremal case \cite{famo1997Cstar} that any $\Ae$-extremal POVM is extremal.

\end{ex}

\begin{ex}\label{ex:cstar} If $\Ae=B(\Ha)$, then $\Ae$-convexity is the same as C*-convexity. In the context of C*-algebras, this notion of convexity, along with the related extremality properties, was studied in \cite{lopa1981some, hmp1981Cstar,famo1993Cstar,morenz1994thestructure}. C*-convexity for sets of generalized states, containing POVMs as a special case, was studied in \cite{fazh1998thestructure,famo1997Cstar}. In particular, it was proved that a POVM  is C*-extremal 
in $\Me(\Ha,n)$ if and only if it is projection-valued. Though it is not clear in general whether C*-extremality implies $\Ae$-extremality, Corollary \ref{coro:PVM} below shows that a projection valued measure (PVM) is $\Ae$-extremal, for any $\Ae\subseteq B(\Ha)$.

\end{ex}


We will also need the notion of an {$\Ae$-face} of $\Me(\Ha,n)$. This is defined as a subset $\Fe$ such that whenever 
$\Fe$ contains  a proper $\Ae$-convex combination of some elements $N^1,\dots,N^k\in \Me(\Ha,n)$, then also $N^j\in \Fe$ for all $j$, see 
\cite{morenz1994thestructure} for a definition of a C*-face of a C*-convex set.  Note that, just as in the case of a C*-face, an $\Ae$-face does not have to be $\Ae$-convex, or even convex. 

\begin{lemma}\label{lemma:extreme_face} Let $M\in \Me(\Ha,n)$ be $\Ae$-extremal then 
\[
\Fe_M=\{U^*MU, U\mbox{ is a unitary in } \Ae\}
\]
is a compact $\Ae$-face of $\Ce$.
\end{lemma}

\begin{proof} It is clear that $\Fe_A$ is an $\Ae$-face. Compactness follows from compactness of the unitary group in $\Ae$.

\end{proof}

The proofs of the lemmas below can be found in  the Appendix.

\begin{lemma}\label{lemma:simored_face} Let $\Fe\subset \Me(\Ha,n)$ be a compact $\Ae$-face. Let $M=\sum_jX_j^*N^jX_j\in \Fe$ be an $\Ae$-convex combination. Then for all $j$, there is some $L^j\in \Me(\Ha,n)$ such that 
$Q_jN^jQ_j+Q_j^\perp L^jQ_j^\perp \in \Fe$, where $Q_j$ is the range projection of $X_j$.

\end{lemma}

We next characterize the $\Ae$-extremal  POVMs.

\begin{lemma}\label{lemma:char_extreme} Let $M\in \Me(\Ha,n)$, then $M$ is $\Ae$-extremal if and only if  $0\le D_i\le M_i$, $i=1,\dots,n$ and  
$\sum_i D_i\in\Ae$ implies that  there is some $X\in \Ae$ such that  $D_i=X^*M_iX$, $i=1,\dots,n$.

\end{lemma}

\begin{coro}\label{coro:PVM} Let $P\in\Me(\Ha,n)$ be a PVM. Then $P$ is $\Ae$-extremal, for any subalgebra $\Ae\subseteq B(\Ha)$.

\end{coro}

\begin{proof}
Let $0\le D_i\le P_i$ be such that $D:=\sum_i D_i\in\Ae$. Then $D$ commutes with all $P_i$ and
\[
D_i=DP_i=D^{1/2}P_iD^{1/2}.
\]
By Lemma \ref{lemma:char_extreme}, $P$ is $\Ae$-extremal.
\end{proof}

We say that an element $M\in \Me(\Ha,n)$ is {$\Ae$-irreducible} if 
the only projections in $\Ae$ commuting with all $M_1,\dots, M_n$  are 0 and $I$.
If  $M$ is $\Ae$-extremal and $\Ae$-irreducible, then $M$ is called {$\Ae$-pure} (cf. \cite{famo1997Cstar}).

\begin{lemma}\label{lemma:char_pure}  Let $M\in \Me(\Ha,n)$ and let $P_i$ be the support projections of $M_i$, $i=1,\dots,n$. Then $M$ is
$\Ae$-pure if and only if  $D_i\in B(P_i\Ha)$,  $\sum_i D_i\in \Ae$, $i=1,\dots,n$ implies that there is some $z\in \mathbb C$ such that $D_i=zM_i$, $i=1,\dots,n$.

\end{lemma}

\section{Process POVMs}\label{sec:ppovm}

A measurement on quantum channels with outcomes in the set $X=\{1,\dots,n\}$ is naturally defined as an affine map $\mathbf m:\Ce(\Ha,\Ka)\to \prob(X)$, the set of  probability distributions over $X$.  
For $\Phi\in \Ce(\Ha,\Ka)$ and  $i\in X$,  the value $\mathbf m(\Phi)_i$ is interpreted as the probability that the outcome of the measurement is $i$ if the true channel is $\Phi$. Similarly to usual quantum measurements, there is a collection of positive 
operators $F_1,\dots,F_n$ associated with $\mathbf m$ \cite{jencova2012generalized}, but acting on the tensor product $\Ka\otimes \Ha$ and with the normalization
$\sum_i F_i=I_\Ka\otimes \sigma$, for some $\sigma\in \states(\Ha)$. The relation of $F$ and $\mathbf m$ is 
\[
\mathbf m(\Phi)_i=\Tr F_i C(\Phi),\qquad i=1,\dots,n.
\]
Any collection $\{F_1,\dots,F_n\}$ of positive operators with this property 
is called a process POVM (see \cite{ziman2008PPOVM}). Moreover, it is easy to see that any process POVM defines a measurement on channels, in the above sense.
We will denote the set of all process POVMs on $\Ka\otimes \Ha$ with $n$ outcomes by $\Fe(\Ha,\Ka,n)$.

 To save some space and simplify notations, we  identify $B(\Ha)$ with the subalgebra $I_\Ka\otimes B(\Ha)\subseteq B(\Ka\otimes \Ha)$. Similarly, operators $T:\Ha\to \Ha_0$ will be identified with their natural extensions $I_\Ka\otimes T: \Ka\otimes \Ha\to \Ka\otimes \Ha_0$.

An obvious way how to perform a measurement on channels is to apply the channel to an input state (possibly on the system coupled with an ancilla) and measure the outcome. The next proposition shows that indeed all measurements are obtained in this way. Moreover, it gives a certain representation result for process POVMs.

\begin{prop}\cite{ziman2008PPOVM}\label{lemma:ppovms_testers}
Let $(\Ha_0,\rho,M)$ be a triple consisting of a Hilbert space $\Ha_0$, a state $\rho\in \states(\Ha\otimes \Ha_0)$ and a POVM $M\in \Me(\Ka\otimes \Ha,n)$. Then there is a unique process POVM $F\in \Fe(\Ha,\Ka,n)$ such that 
\begin{equation}\label{eq:ppovm_tester}
\Tr C(\Phi)F_i=\Tr  M_i(\Phi\otimes id_{\Ha_0})(\rho),\quad  i=1,\dots,n, \ \Phi\in \Le(\Ha,\Ka).
\end{equation}
Conversely, for any $F\in \Fe(\Ha,\Ka,n)$, there exists a triple $(\Ha_0,\rho,M)$ with a pure state $\rho$,  such that  (\ref{eq:ppovm_tester}) holds.
\end{prop}

\begin{proof}
Let $(\Ha_0,\rho,M)$ be such a triple. By the Choi isomorphism, there exists a completely positive map $\Phi_\rho: B(\Ha)\to B(\Ha_0)$ such that 
$(id_\Ha\otimes \Phi_\rho)(\psi_\Ha)=\rho$. This map is given by
\[
\Phi_\rho(A)= \ptr_\Ha [\rho A^t],\qquad A\in B(\Ha),
\]
where $A^t$ is the transpose of $A$ with respect to the ONB $\{|i\>\}$. For $\Phi\in \Le(\Ha,\Ka)$, 
\[
\Tr M_i (\Phi\otimes id_{\Ha_0})(\rho)=\Tr M_i(\Phi\otimes \Phi_\rho)(\psi_\Ha)=\Tr F_iC(\Phi),
\]
with 
\[
F_i=(id\otimes \Phi_\rho^*)(M_i),
\]
where  $\Phi_\rho^*: B(\Ha_0)\to B(\Ha)$ is the adjoint of $\Phi_\rho$ with respect to the Hilbert-Schmidt
 inner product, that is,
\begin{equation}\label{eq:R}
\Phi_\rho^*(B)=\left( \ptr_{\Ha_0} [B\rho]\right)^t,\qquad B\in B(\Ha_0).
\end{equation}
 Since $0\le M_i\le I$, $\sum_i M_i=I$  and 
$\Phi_\rho^*$ is completely positive, we have $F_i\ge 0$ and $\sum_i F_i= \sigma$, where 
\[
\sigma:=\Phi_\rho^*(I_{\Ha_0})=(\ptr_{\Ha_0}\rho)^t\in \states(\Ha). 
\]
Hence $F=\{F_1,\dots,F_n\}$ is a process POVM.
 Uniqueness follows from the fact that $C$ is an isomorphism $\Le(\Ha,\Ka)$ onto $B(\Ka\otimes \Ha)$.

Conversely, let $F$ be a process POVM, $\sum_iF_i= \sigma$, $\sigma\in \states(\Ha)$. Let $\supp(\sigma)$ be the projection onto the support of $\sigma$ and let $\Ha_0=\supp(\sigma)\Ha$. Then $\supp(F_i)\le  \supp(\sigma)$
 and we may put $M_i:= \sigma^{-1/2}F_i \sigma^{-1/2}\in B(\Ka\otimes \Ha_0)$. 
Clearly, $M_i\ge 0$ and $\sum_i M_i =I_{\Ka\otimes \Ha_0}$, so that $M\in \Me(\Ka\otimes \Ha_0,n)$. Moreover, for $\Phi\in \Le(\Ha,\Ka)$,
\[
\Tr F_iC(\Phi)=\Tr M_i \sigma^{1/2}C(\Phi)\sigma^{1/2}=\Tr M_i(\Phi\otimes id)(\rho),
\] 
where 
\[
\rho= \sigma^{1/2}\psi_H \sigma^{1/2}
\]
is a pure state in $\states(\Ha\otimes \Ha_0)$. Then $(\Ha_0,\rho,M)$ is the required triple with a pure input state.

\end{proof}

Any triple satisfying (\ref{eq:ppovm_tester})  will be called a 
{representation} of $F$.  
If $(\Ha_0,\rho,M)$ and $(\Ha_0',\rho',M')$ are representations of the same process POVM $F$, we will say that the triples are equivalent and write
\[
(\Ha_0,\rho,M)\simeq (\Ha_0',\rho',M')\simeq F.
\]
The proof of the following lemma is straightforward.

\begin{lemma}\label{lemma:approx} Let $\rho\in \states(\Ha\otimes \Ha_0)$ and  $M'\in \Me(\Ka\otimes \Ha_0',n)$. Then for any 
$\chi\in \Ce(\Ha_0,\Ha_0')$, 
\[
(\Ha_0',(id\otimes \chi)(\rho),M')\simeq (\Ha_0,\rho, (id\otimes \chi^*)(M')).
\]
\end{lemma}

\subsection{Minimal representations}

Let $(\Ha_0,\rho,M)\simeq F$ and assume the input state $\rho$ is pure. Then $\rho=|T\kk\bb T|$ and 
$\Phi_\rho^*(a)=T^*aT$ for some $T:\Ha\to \Ha_0$. By the proof of Lemma \ref{lemma:ppovms_testers}, 
\begin{equation}\label{eq:pure}
F_j= T^*M_j T, \quad j=1,\dots,n.
\end{equation}
We say the representation (or the triple $(\Ha_0,\rho,M)$) is {minimal} if the input state is pure and $T$ is surjective.  It is clear
 that the representation constructed in the proof of 
 Lemma \ref{lemma:ppovms_testers} is minimal.

\begin{lemma}\label{lemma:minimal_testers} Let $(\Ha_0,\rho,M)$ be  minimal  and let   
$(\Ha_0',\rho',M')\simeq (\Ha_0,\rho,M)$. 
Then there is a channel  $\chi: B(\Ha_0)\to B(\Ha_0')$, such that $(id_\Ha\otimes \chi)(\rho)=\rho'$ and $M=(id_\Ha\otimes \chi^*)(M')$. 
If if $\rho'$ is pure, then $\chi$ is an isometric channel  and if $(\Ha_0',\rho',M')$ is minimal, then $\chi$ is a unitary channel.

\end{lemma}

\begin{proof} By minimality, $\rho=|T\kk\bb T|$ with $T:\Ha\to \Ha_0$ surjective.
Suppose first that $\rho'$ is pure, $\rho'=|T'\kk\bb T'|$. Since the two triples are representations  of the same process POVM $F$, we must have 
\[
T^*T=\Phi_\rho^*(I_{\Ha_0})=\sum_jF_j=\Phi_{\rho'}^*(I_{\Ha_0'})=(T')^*T'. 
\]
 By considering polar decompositions of $T$ and $T'$, and using the fact that $T$ is surjective, we obtain that there is some isometry $U:\Ha_0\to \Ha_0'$ such that $T'=UT$, so that
\[
\rho'= T'\psi_\Ha (T')^*= U\rho U^*=(id_\Ha\otimes Ad_{U})(\rho).
\]
Similarly,
\[
F_i= T^*M_i T= (T')^*M_i' T'= T^*U^*M_i' UT
\]
and this implies $M=(id_H\otimes Ad_{U^*})(M')$ by surjectivity of $T$. If $(\Ha_0',\rho',M')$ is minimal, then also $T'$ is surjective and this implies that $U$ is  unitary.

Let now $\rho'$ be any state and let $\rho'_0$ be its purification, that is, there is a Hilbert space $\Ha_1$ such that 
$\rho_0'\in \states(\Ha\otimes \Ha_0'\otimes \Ha_1)$ is a pure state and $\ptr_{\Ha_1}\rho'_0=\rho'$. By  Lemma \ref{lemma:approx}, we obtain
 \[
(\Ha_0,\rho, M)\simeq (\Ha_0',\rho',M')=(\Ha_0',\ptr_{\Ha_1}(\rho_0'),M')\simeq (\Ha_0'\otimes \Ha_1, \rho_0',M').
 \]
By the first part of the proof, there is some isometry 
$U: \Ha_0\to \Ha_0'\otimes \Ha_1$ such that $\rho_0'= U\rho U^*$ and
 $ U^*M' U=M$. Put $\chi=\ptr_{\Ha_1}\circ Ad_{U}$, then $\chi$ is a channel 
 $B(\Ha_0)\to B(\Ha_0')$ and we have $\rho'=\ptr_{\Ha_1}\rho_0'=(id_\Ha\otimes \chi)(\rho)$, $M=(id_\Ha\otimes \chi^*)(M')$. 

\end{proof}

%
%
%
%
%
%

Let $F$ be a process POVM and let  $\sum_j F_j= \sigma\in\states(\Ha)$. Then 
$r(F):=\mathrm{rank}(\sigma)$ will be called the {rank} of $F$.

\begin{lemma}\label{lemma:rank}
Let $F$ be a process POVM with rank $r$. Then for any representation $(\Ha_0,\rho, M)\simeq F$, the Schmidt number of the input state
$SN(\rho)\le r$. If $\rho$ is pure, then $SN(\rho)=r$ and the representation is minimal if and only if $\dim(\Ha_0)=r$.

\end{lemma}

\begin{proof} By the proof of Lemma \ref{lemma:ppovms_testers}, there is a minimal
representation $(\Ha_0,\rho,M)\simeq F$, where 
$\Ha_0$ is the range of $\sigma$, hence $SN(\rho)=\dim(\Ha_0)=r(F)$. The rest follows by Lemma \ref{lemma:minimal_testers}.

\end{proof}

\begin{rem}\label{rem:minimal} If the input state is pure, we will often abuse the notation and write the triple as $(\Ha_0,T,M)$, instead of $(\Ha_0, |T\kk\bb T|,M)$.
Let $(\Ha_0,T,M)\simeq F$ and let $\Ha_T\subseteq \Ha_0$ be the range of $T$, with $P_T$ the projection onto $\Ha_T$ (we will keep this notation throughout the paper). Then it is easy to see that $(\Ha_T,T, P_TM P_T)$ is an equivalent  minimal representation. 
On the other hand, if $\Ha_0\subseteq \Ha_0'$, then $F\simeq (\Ha_0',T,\widetilde M)$, where  $\widetilde M$ is the extension of $M$  in $\Me(\Ka\otimes \Ha_0',n)$ as in Remark \ref{rem:ext}.

\end{rem}

\section{Extreme points and faces of $\Fe(\Ha,\Ka,n)$}\label{sec:extr}

It is easy to see that the set of process POVMs is convex and compact.  We will now examine the extremal points and faces of the set $\Fe(\Ha,\Ka,n)$ and obtain a description  in terms of the  minimal representing triples. As it turns out, the convex structure of $\Fe(\Ha,\Ka,n)$ is closely related to the $B(\Ha_0)$-convex structure of $\Me(\Ka\otimes \Ha_0,n)$, where  $B(\Ha_0)$ is seen as a subalgebra of $B(\Ka\otimes \Ha_0)$. 

We first look at representing triples  of convex combinations of process POVMs.
 Let $F^1,\dots, F^m$ be process POVMs and let $(\Ha_i,T_i, M^i)\simeq F_i$ be representations (not necessarily minimal) with pure input states. Consider a convex combination $F=\sum\lambda_i F^i$, with $\lambda_i>0$ for all $i$.
Put 
\begin{equation}\label{eq:S}
S:=(\sum_i\lambda_i T_i^*T_i)^{1/2}\\
\end{equation}
and define the maps $X_i:\Ha_S\to \Ha_i$ by 
\begin{equation}\label{eq:Xi}
X_i\xi=\lambda_i^{1/2}T_iS^{-1}\xi, \qquad \xi\in \Ha_S,
\end{equation}
where $S^{-1}$ is the inverse of $S$ on $\Ha_S$ and 0 on the complement. Let also
\begin{equation}\label{eq:M}
N=\sum_i  X_i^* M^i X_i.
\end{equation}
Since $\sum_iX_i^*X_i=I_{\Ha_S}$, $N$ is a POVM on $\Ka\otimes \Ha_S$ and we have
\[
 S^*N S= \sum_i \lambda_i  T_i^*M^i T_i=\sum_i\lambda_iF^i=F.
\]
Hence $F\simeq (\Ha_S,S,N)$ and this representation is minimal. On the other hand, let $F$ be a process POVM and let 
$F\simeq (\Ha_0,T,M)$ be a representation, with a pure input state but again minimality is not assumed. Suppose that 
$M=\sum_{i=1}^m X_i^*M^i X_i$ is a $ B(\Ha_0)$-convex combination of some $M^i\in \Me(\Ka\otimes \Ha_0,n)$.
Put
\begin{equation}\label{eq:I}
\mathcal I:=\{i\in \{1,\dots,m\},\ X_iT\ne 0\}
\end{equation}
 and for $i\in \mathcal I$, we define
\begin{equation}\label{eq:muTiFi}
\mu_i:=\Tr T^*X_{i}^*X_{i}T, \qquad T_i:=\mu_i^{-1/2}X_{i}T,
\qquad F^i \simeq (\Ha_0,T_i,M^{i}).
\end{equation}
Then it is easy to see that $\mu_i>0$, $\sum_i\mu_i=1$, $F_i$ are process POVMs and
\[
F=T^*MT=\sum_iT^*X_i^*M^iX_iT=\sum_{i\in \mathcal I}\mu_i  T_i^*M^i T_i=\sum_{i\in \mathcal I} \mu_iF_j^i.
\]

Our first main result is the following characterization of minimal representations of extremal process POVMs.

\begin{theorem}\label{thm:extreme} Let $F\in \Fe(\Ha,\Ka,n)$ and let $(\Ha_0,T,M)\simeq F$ be a minimal representation. Then the following statements are equivalent.
\begin{enumerate}
\item[(i)] $F$ is extremal.
\item[(ii)] $M$ is  $ B(\Ha_0)$-extremal and for any representation
$(\Ha_0',\rho',M')\simeq F$, we have $SN(\rho')=r(F)$.
\item[(iii)] $M$ is  $ B(\Ha_0)$-pure.

\end{enumerate}

\end{theorem}

\begin{proof} Suppose (i) and let $(\Ha_0',\rho',M')\simeq F$. By Lemma \ref{lemma:rank}, $SN(\rho')\le r(F)$, with equality if $\rho'$ is pure. Let  $\rho'=\sum_i\lambda_i \rho_i$ for some $0<\lambda_i<1$,
 $\sum_i\lambda_i=1$ and pure states $\rho_i$, and assume that $SR(\rho_k)<r(F)$ for some $k$. Then since $\rho_k$ is pure, $(\Ha_0',\rho_k,M')$ cannot be a representation of $F$.  Let $F^i$ be process POVMs such that $(\Ha_0',\rho_i,M')\simeq F^i$, then  $F^k\ne F$. But then 
 $F=\sum_i \lambda_i F_i$, with $F_k\ne F$, a contradiction. Hence we must have $SN(\rho')=SR(\rho_k)=r(F)$, for all $k$. 

Further, suppose $M=\sum_i X_i^*M^i X_i$ is a proper $ B(\Ha_0)$-convex combination of some 
$M^i\in \Me(\Ka\otimes \Ha_0,n)$. Note that then $X_iT\ne 0$ for all $i$ and  $F=\sum_i\mu_iF^i$, where $F^i\simeq(\Ha_0,T_i,M^i)$ are as in (\ref{eq:muTiFi}). Since $X_i$ is invertible and $T$ is surjective, 
these are minimal representations of $F^i$. By extremality, we must have $F^i=F$, so that $(\Ha_0,T_i, M^i)\simeq F$ for all $i$ 
and by Lemma \ref{lemma:minimal_testers}, this implies 
that $M^i\sim_{ B(\Ha_0)} M$. Hence $M$ is $ B(\Ha_0)$-extremal, this proves (ii).

Suppose (ii). We have to show that $M$ is $ B(\Ha_0)$-irreducible. So let $P\in B(\Ha_0)$ be a projection such that
 $ PM_j=M_jP$ for all $j$. Let $\chi_P:B(\Ha_0)\ni A\mapsto PAP+P^\perp AP^\perp$, then $\chi_P$ 
 is a unital channel such that $\chi_P^*=\chi_P$ and $(id\otimes \chi_P)(M)=M$. By Lemma \ref{lemma:approx}, 
 $F\simeq (\Ha_0,(id\otimes \chi_P)(|T\kk\bb T|),M)$ and it is clear that $SN((id\otimes \chi_P)(|T\kk\bb T|))=r(F)=\mathrm{rank}(T)$ if and only if $P=0$ or $I$. This proves (iii). 

Finally, suppose (iii) and let $F=\sum_i\lambda_iF^i$, where $F^i\in \Fe(\Ha,\Ka,n)$ have minimal minimal representations
$F^i\simeq(\Ha_i,T_i,M^i)$. Then $F$ has a minimal representation $(\Ha_{S},S,N)$ given by equations (\ref{eq:S})-(\ref{eq:M}).
Clearly, we may assume that $\Ha_i\subseteq \Ha_{S}$, so that $X_i\in B(\Ha_{S})$ and $N=\sum_i  X_i^*\widetilde M^i X_i$
 is a $ B(\Ha_{S})$-convex combination of the extensions $\widetilde M_i \in \Me(\Ka\otimes \Ha_{S},n)$.
By Lemma \ref{lemma:minimal_testers}, $M= U^*N U$ for some unitary $U: \Ha_0\to \Ha_{S}$, so that $N$ is $ B(\Ha_{S})$-pure. For all $i$ and $j$, we have  $D^i_j:=X_i^*\widetilde M^i_jX_i\le N_j$ and $\sum_j D^i_j=X_i^*X_i\in B(\Ha_S)$. Lemma \ref{lemma:char_pure} now implies that $D_j^i=s_iN_j$ for all $j$ and some $s_i\in [0,1]$. But then we must have
\[
F^i=T_i^*M^iT_i=S^*NS=F,\qquad \forall i
\]
and $F$ is extremal.

\end{proof}

\subsection{Faces of $\Fe(\Ha,\Ka,n)$}

 Let $\Ha_0\subseteq \Ha$ be a subspace. Then any $F\in \Fe(\Ha,\Ka,n)$ with $r(F)\le \dim(\Ha_0)=:r$ has a representation with ancilla $\Ha_0$ and a pure input state. Once the ancilla is fixed, any such $F$ can be represented by a pair $(T,M)$, where 
\[ 
T\in \Te:=\{T:\Ha\to \Ha_0, \Tr T^*T=1\}
\]
 and $M\in \Me:=\Me(\Ka\otimes \Ha_0,n)$. Conversely, for any $(T,M)\in \Te\times \Me$, 
 $(\Ha_0,T,M)\simeq F$ with $r(F)\le r$. Equivalence of representations defines an equivalence relation on $\Te\times\Me$ 
 as $(T,M)\simeq (T'M')$ if and only if  $T^*MT=(T')^*M'T'$. Note that
 $(T,M)\simeq (T',M')$ if and only if there is a unitary $U\in B(\Ha_0)$  such that
  \begin{equation}\label{eq:siminTxM}
T'=UT\ \mbox{ and } \  P_TM P_T= P_TU^*M' UP_T.
  \end{equation}

Let $\Fe\subseteq \Fe(\Ha,\Ka,n)$ and let $r:=\sup_{F\in \Fe} r(F)$. Fix a subspace $\Ha_0\subseteq \Ha$ with $\dim(H_0)=r$ and let $\Te$ and $\Me$ be as above.  Let $S_\Fe=\{(T,M)\in\Te\times\Me,\ T^*MT\in \Fe\}$ 
and let $\Me_\Fe$ and $\Te_\Fe$ be the projections of $S_\Fe$ in $\Me$ and $\Te$.
We also denote by $\Ha_\Fe$ the subspace in $\Ha$ generated by the ranges of $T^*$ for $T\in \Te_\Fe$.

Note that if $\Fe$ is closed, then $S_\Fe$, $\Me_\Fe$ and $\Te_\Fe$ are compact. Indeed, since both $\Me$ and $\Te$ are compact sets, so is the product $\Te\times \Me$. The map 
$(T,M)\mapsto  T^*M T$ is  continuous and $S_\Fe$ is the pre-image of $\Fe$, so that $S_\Fe$ is a closed subset of $\Te\times \Me$. It follows that $S_\Fe$ is compact and so are the projections  $\Me_\Fe$ and $\Te_\Fe$.

%
%
%

\begin{theorem}\label{thm:faces} Let $\Fe\subseteq \Fe(\Ha,\Ka,n)$ and let $\sup_{F\in \Fe} r(F)=r$. Then $\Fe$ is a face of $\Fe(\Ha,\Ka,n)$ if and only if the following conditions hold:
\begin{enumerate}
\item[(i)] $\dim(\Ha_\Fe)=r$. 
\item[(ii)] $\Me_\Fe$  is an $ B(\Ha_0)$-face of $\Me$.
\item[(iii)] Let $M^1,\dots, M^m\in \Me_\Fe$, $X_1,\dots,X_m\in B(\Ha_0)$ be such that $\sum_iX_i^*X_i=I$ and 
$T\in \Te$. Then $(T,\sum_i X_i^*M^i X_i)\in S_\Fe$ if and only if $(\mu_{X_i}^{-1/2}X_iT,M^i)\in S_\Fe$ for all $i$ such that  $\mu_{X_i}^2:=\Tr T^*X_i^*X_iT\ne 0$.

%
%
\end{enumerate}

\end{theorem}

\begin{rem} Recall that a $B(\Ha_0)$-face is not necessarily $B(\Ha_0)$-convex. The condition (iii) in the above theorem specifies those $B(\Ha_0)$-convex combinations of elements of the face $\Me_\Fe$ that it contains. 
\end{rem}

\begin{proof} Suppose $\Fe$ is a face. The condition (i) is a consequence of convexity of $\Fe$. Indeed, let $F$ be an interior point of $\Fe$ and let
$(T,M)$ be a corresponding pair, then $\sum_jF_j= T^*T=:\sigma\in \states(\Ha)$. Let $T'$ be any  element in $\Te_\Fe$. Then there is some $M'$ such that $(H_0,T',M')\simeq F'\in \Fe$ and $(1+s)F-sF'\in \Fe$ for some $s>0$. 
This implies  $(1+s)\sigma-s\sigma'\in\states(\Ha)$, where $\sigma'=(T')^*T'$.
Consequently, the support of $\sigma'$, which is the same as the range of $(T')^*$, is contained in the support of $\sigma$. It follows that the support of $\sigma$ is equal to $\Ha_\Fe$ and (i) holds.

To prove (ii), let $M=\sum_i X_i^*M^i X_i\in\Me_\Fe$ be a proper $ B(\Ha_0)$-convex combination of $M^i\in \Me$. Then there is some $T\in \Te$ such that $(\Ha_0,T,M)\simeq F\in \Fe$ and $F=\sum_{i\in \mathcal I}\mu_i F^i$, where $\mathcal I$ and 
$F^i\simeq (\Ha_0,T_i,M^i)$ are  as in (\ref{eq:I}) and
 (\ref{eq:muTiFi}). Since $\Fe$ is a face, $F^i\in \Fe$ and hence $M^i\in \Me_\Fe$ (note that all $i$ are in $\mathcal I$ since $X_i$ is invertible). It follows that $\Me_\Fe$ is a $ B(\Ha_0)$-face of $\Me$.

For (iii), let $\sum_i  X_i^*M^i X_i=M$ and let $T\in\Te$ be such that $(t_iX_iT,M^i)\in S_\Fe$ for all $i\in \mathcal I$. Let $F^i\simeq (\Ha_0,t_iX_iT,M^i)$ for $i\in \mathcal I$, then $F^i\in \Fe$. Put $\mu_i:=t_i^{-2}=\Tr T^*X_i^*X_iT$, then $\sum_{i\in\mathcal I} \mu_i=1$ and since $\Fe$ is convex, we have $\sum_{i\in \mathcal I} \mu_i F^i\in \Fe$ and
 \[
\sum_i\mu_i F_i=\sum_i \mu_it_i^2T^*X_i^*M^iX_iT=T^*MT
 \]
so that $(T,M)\in S_\Fe$. For the converse of (iii),
let $T\in \Te$ be such that $T^*MT=F\in \Fe$. Then $F=\sum_{i\in \mathcal I}\mu_i F^i$ as in (\ref{eq:muTiFi}). Since $\Fe$ is a face, this 
 implies $F^i\in \Fe$ and $(\mu_{i}^{-1/2}X_{i}T,M^{i})\in S_\Fe$ for all $i\in \mathcal I$.

Conversely, assume (i)-(iii) hold. Let $F^1,F^2\in \Fe$, $(H_0,T_i,M^i)\simeq F^i$ and let  $\lambda F^1+(1-\lambda)F^2=F$ for some $\lambda\in (0,1)$. 
Then $F\simeq(H_{S},S,N)$ as in (\ref{eq:S})-(\ref{eq:M}). We have $H_{S}\subseteq H_\Fe$ and by (i), there is a unitary  $U:H_\Fe\to H_0$. Put $T:=US\in \Te$ and
\[
Y_i:=X_iU^*,\ i=1,2,\quad Y_{3}=I_{H_0}-P_T
\]
then $Y_i\in B(H_0)$, $\sum_i Y_i^*Y_i=I$ and 
\[
\lambda^{-1/2}Y_1T=T_1,\quad (1-\lambda)^{-1/2}Y_2T=T_2, \quad  Y_3T=0.
\]
Put $M=\sum_i Y_i^*M^i Y_i$ where $M^3$ is any element in $\Me_\Fe$, then  it follows by (iii) that  $(T, M)\in S_\Fe$. Since clearly $T^*MT=S^*NS$, we have
 $F\simeq (H_{S},S,N)\simeq (H_0, T, M)$ is in $\Fe$, so that $\Fe$ is convex.

To prove that $\Fe$ is a face, let $F^1,F^2\in \Fe(\Ha,\Ka,n)$ and $\lambda\in (0,1)$ be such that 
$\lambda F^1+(1-\lambda)F^2=F\in \Fe$.  It is clear that then $r(F^i)\le r(F)\le r$, so that there are some $(T_i,M^i)\in \Te\times \Me$ corresponding to $F^i$. Then $F$ has a minimal representation $(\Ha_{S},S,N)$ as in  (\ref{eq:S})-(\ref{eq:M}) and since  $F\in \Fe$, $F\simeq (\Ha_0,T,M)$ for some $(T,M)\in S_\Fe$. By Lemma 
\ref{lemma:minimal_testers}, there is an isometry $V:\Ha_{S}\to \Ha_0$ such that $T=VS$ and $N= V^*M V$. 
Exactly as before, we put $Y_i=X_iV^*$ for $i=1,2$ and $Y_3=I-P_T$. Then $Y_1,Y_2,Y_3\in B(\Ha_0)$,  $\sum_iY_i^*Y_i=I_{\Ha_0}$ 
and 
\[
M':=\sum_{i=1}^3 Y^*_iM^i Y_i,
\]
with arbitrary $M^3\in \Me_\Fe$  satisfies
$T^*M'T=S^*NS$,
hence $(T,M')\in S_\Fe$. By (ii) and Lemmas \ref{lemma:nto2} and \ref{lemma:simored_face}, there are elements 
$R^i\in \Me$, $i=1,2$ such that 
\[
N^i:= Q_iM^iQ_i+Q_i^\perp R^iQ_i^\perp\in \Me_\Fe
\]
where $Q_i$ is the range projection of $Y_i$. Then with $N^3=M^3$, we have 
\[
M'=\sum_{i=1}^3 Y_i^*N^i Y_i
\]
By (iii),  $(t_iY_iT,N^i)\in S_\Fe$, with
\[
t_i^{-2}=\Tr T^*Y_i^*Y_iT= \Tr S^*X_i^*X_iS=\lambda_i \Tr T^*_iT_i=\lambda_i,
\]
where $\lambda_1=\lambda$, $\lambda_2=1-\lambda$. Moreover, 
\[
t_i^2T^*Y_i^*N^iY_iT=t_i^{-2}T^*Y_i^*M^iY_iT=t_i^{-2}SX_i^*M^iX_iS=T_i^*M^iT_i=F^i.
\]
It follows that $F^i\in \Fe$, hence $\Fe$ is a face of $\Fe(\Ha,\Ka,n)$.

\end{proof}
\subsection{The face generated by a process POVM with $B(\Ha_0)$-extremal measurement}

Let $M$ be a $B(\Ha_0)$-extremal POVM on $\Ka\otimes \Ha_0$  and let $T:\Ha\to \Ha_0$ be a surjective linear map. Then $(\Ha_0,T,M)$ is a minimal representation of some process POVM $F=T^*MT$. As an illustration of our results, we will describe the face of $\Fe(\Ha,\Ka,n)$, generated by $F$. 

Let $M$ be any element of $\Me(\Ha,n)$. A Naimark representation of $M$ is a triple $(\tilde \Ha, E,J)$, where $\tilde \Ha$ is some Hilbert space,   $E\in \Me(\tilde \Ha,n)$ is projection valued and 
$J:\Ha\to \tilde \Ha$ is an isometry such that $M=J^*EJ$. If moreover the set 
$\{E_kJ\xi,\ k=1,\dots,n, \xi\in\Ha\}$ spans $\tilde\Ha$, such a representation is called minimal. In fact,
$(c_1,\dots,c_n)\mapsto \sum_j c_j M_j$ defines a completely positive map $\mathbb  C^n\to B(\Ha)$ and $(\tilde \Ha, E, J)$ is a (minimal) Stinespring representation of this map, see for example \cite{paulsen2002completely}.

We will denote by $\{E\}'$  the commutant of $\{E_1,\dots,E_n\}$ in $B(\tilde \Ha)$.

\begin{lemma}\label{lemma:L_M}  Let $M\in \Me(\Ha,n)$, with minimal Naimark representation $M=J^*EJ$, and let $X\in B(\Ha)$. Then
$X^*MX\le M$ if and only if there is some $C\in \{E\}'$, $\|C\|\le 1$, such that $CJ=JX$.

\end{lemma}

\begin{proof} Note that  $\Phi_X: (c_1,\dots,c_n)\mapsto \sum_j c_jX^*M_jX$ defines a completely positive map $\mathbb C^n\to B(\Ha)$ with a Stinespring representation $(\tilde H, E, JX)$ and the condition $X^*MX\le M$ is equivalent to a natural ordering of the corresponding maps. By the 
Radon-Nikodym theorem for completely positive maps (\cite[Theorem 1.4.2.]{arveson1969I}), there is some element $0\le B\le I$ in $\{E\}'$, such that 
\[
X^*M_jX=J^*E_jBJ=J^*B_jJ,
\]
where $B_j=E_jB$.  Let  $Q_j$ denote the support projection of 
$B_j$, then  $Q=\{Q_1,\dots, Q_n\}$ is a projection valued measure on the range 
$\tilde \Ha_B$ of $B$. 
 It follows that
$(\tilde \Ha_B,Q,B^{1/2}J)$ is another Stinespring representation of $\Phi_X$, which is obviously minimal. It follows that there is a partial isometry $W\in B(\tilde \Ha)$ with initial space $\tilde \Ha_B$, 
such that $WB^{1/2}J=JX$ and
$W^*E_jW=Q_j\le E_j$ for all $j$. This implies that $W\in \{E\}'$. Put $C:=WB^{1/2}$, then $C\in \{E\}'$, $\|C\|\le 1$ and $CJ=JX$.
Conversely, if $JX=CJ$, for $C$ as above, then for all $j$, 
\[
X^*M_jX=X^*J^*E_jJX=J^*C^*E_jCJ=J^*C^*CE_jJ\le J^*E_jJ=M_j.
\]
\end{proof}

For $M\in \Me$, we will  denote $\mathcal L_M:=\{X\in B(\Ha_0),\ X^*MX\le M\}$. It is easy to see by Lemma \ref{lemma:L_M} that $\Le_M$ is a subalgebra (but not necessarily a *-subalgebra) in $B(\Ha_0)$.  If $M$ is $B(\Ha_0)$-extremal 
in $\Me$ and $X\in \Le_M$, then 
$0\le M_j-X^*M_jX\le M_j$ for all $j$, so that by Lemma \ref{lemma:char_extreme} there is some $Y\in B(\Ha_0)$ such that  $M_j-X^*M_jX=Y^*M_jY$, $\forall j$. It is clear that $X^*MX+Y^*MY=M$, $X^*X+Y^*Y=I$ and $Y\in \Le_M$.  

Let us now fix some $B(\Ha_0)$-extremal element $M\in \Me$ and some surjective $T\in \Te$. Let
\[
\Fe_{T,M}:=\{ \mu^{-1}_X T^*X^*MXT,\ X\in \mathcal L_M,\ \mu_X:=(\Tr T^*X^*XT)\ne 0\}.
\]
In the rest of this section, we will prove that $\Fe_{T,M}$ is the smallest face of $\Fe(\Ha,\Ka,n)$ containing  
$F=T^*MT$. We will first find the projection  $\Me_{T,M}:=\Me_{\Fe_{T,M}}$ of $\Fe_{T,M}$ in $\Me$.

\begin{lemma}\label{lemma:M_F} $\Me_{T,M}$ is the set of all $N\in \Me$, such that there exists a projection  $Q\in B(\Ha_0)$ and  a unitary $U\in B(\Ha_0)$, satisfying
\[
 QU^*NUQ= QM=M Q.
\]

\end{lemma}
\begin{proof}
Assume $N\in \Me_{T,M}$, then there is some $S\in \Te$ and $X\in \mathcal L_M$ such that $(S,N)\simeq (\mu_X^{-1/2} XT,M)$. Hence by (\ref{eq:siminTxM}), there is 
 a unitary $V\in B(\Ha_0)$ such that $S=\mu_X^{-1/2} VXT$ and $ P_XMP_X= P_XV^*N VP_X$, note that since $T$ is surjective $P_X=P_{XT}$.
  Let $Y\in \mathcal L_M$ be such that $X^*M X+ Y^*M Y=M$. This is a $B(\Ha_0)$-convex combination, so by 
 lemmas \ref{lemma:extreme_face} and \ref{lemma:simored_face}, we obtain that 
\[
W^*MW= P_XM P_X+ P_X^\perp R P_X^\perp,
\]
for some unitary $W\in B(\Ha_0)$ and some $R\in \Me$. Put $Q:=WP_XW^*$, $U:=VW^*$ and $R'=WRW^*$. Then we have
\[
 M=QU^*N UQ+Q^\perp R'Q^\perp,
\]
so that $ QU^*N UQ= QM=MQ$.
 
 Conversely, let $N$, $Q$ and $U$ be as in the lemma, then clearly $Q\in \mathcal L_M$ and $(\mu_Q^{-1/2} QT,M)\simeq (\mu_Q^{-1/2} UQT,N)$ defines an element in $\Fe_{T,M}$, so that $N\in \Me_{T,M}$. 
\end{proof}

\begin{lemma} $\Me_{T,M}$ is a $ B(\Ha_0)$-face of $\Me$.

\end{lemma}

\begin{proof} Let $\sum_i  X_i^*M^i X_i=N\in \Me_{T,M}$ be a proper $ B(\Ha_0)$-convex combination of $M^i\in \Me$. Let $Q$ and $U$ be as in Lemma \ref{lemma:M_F}. 
Then $\sum_i  QU^*X_i^*M^i X_iUQ= QM$, so that 
\[
M=\sum_i  Y_i^*M^i Y_i+ Q^\perp M.
\]
By Lemmas  \ref{lemma:extreme_face} and \ref{lemma:simored_face},
for each $i$ there is some $R^i\in \Me$ such that 
\[
W_i^*MW_i= P_{Y_i}M^i P_{Y_i}+ P_{Y_i}^\perp R^i P_{Y_i}^\perp
\]
for some unitary $W\in B(H_0)$. Put $Q_i=WP_{Y_i}W^*$, then
\[
 Q_iWM^i W^*Q_i= Q_iM=M Q_i,
\]
so that $M^i\in \Me_{T,M}$.

\end{proof}

\begin{prop}\label{prop:smallestface} $ \Fe_{T,M}$ is the smallest face of $\Fe(\Ha,\Ka,n)$ containing $F=T^*MT$.

\end{prop}

\begin{proof} Let $\Fe$ be any face containing $F$. Let $X\in \mathcal L_M$ be such that $\mu_X\ne 0$ and let 
$Y\in \Le_M$ be such that  $M= X^*M X+ Y^*M Y$.
By Theorem \ref{thm:faces} (iii), we must have $(\mu_X^{-1/2} XT,M)\in S_\Fe$, so that $\Fe_{T,M}\subseteq \Fe$.
It is now enough to prove that $\Fe_{T,M}$ satisfies the conditions  (i) and (iii) from Theorem \ref{thm:faces}. 

For (i), note that $r=\sup_{F'\in \Fe_{T,M}}r(F')=\dim(\Ha_0)$ and if $(S,N)\in S_{\Fe_{T,M}}$, then $S=\mu_X^{-1/2} UXT$ for some unitary $U\in B(\Ha_0)$ and hence $S^*\Ha_0=T^*X^*U\Ha_0\subseteq T^*\Ha_0$, so that $\Ha_{\Fe_{T,M}}\subseteq T^*\Ha_0\subseteq 
\Ha_{\Fe_{T,M}}$.

To prove (iii), let $M^1,\dots,M^m\in \Me$, $X_1,\dots, X_m\in B(\Ha_0)$, $\sum_i X_i^*X_i=I$ and let $S\in \Te$. Put $N=\sum_i X_i^*M^i X_i$ and assume that  
 $(S,N)\in S_{\Fe_{M,T}}$.  As in the proof of Lemma \ref{lemma:M_F}, there is some $X\in \mathcal L_M$, a unitary $V\in B(\Ha_0)$ and some $M'\sim_{ B(\Ha_0)} M$ such    that $S=\mu_X^{-1/2} VXT$ and 
\[
 P_XV^*N VP_X= P_XM P_X= P_XM'=M' P_X
\]
Substituting for $N$ and putting $Y_i:=X_iVP_X$, we obtain
\[
\sum_i  Y_i^*M^i Y_i+ P_X^\perp M' P_X^\perp=M'
\]
By Lemmas \ref{lemma:extreme_face} and \ref{lemma:simored_face}, for each $i$, there is some $R^i\in \Me$  and a unitary $Z_i\in B(\Ha_0)$ such that
\[
 P_{Y_i}M^i P_{Y_i}+ P_{Y_i}^\perp R^i P_{Y_i}^\perp= Z_i^*MZ_i.
\]
It follows that 
\begin{equation}\label{eq:eq_j}
 P_{Y_i}M^i P_{Y_i} = P_{Y_i}Z_i^*M Z_iP_{Y_i}.
\end{equation}

Now note that $X_iS=\mu_X^{-1/2} X_iVXT=\mu_X^{-1/2} Y_iXT$ and $P_{X_iS}=P_{Y_i}$. If $\mu_i=\Tr S^*X_i^*X_iS\ne 0$, then by (\ref{eq:siminTxM}) and (\ref{eq:eq_j}),
\[
(\mu_i^{-1/2} X_iS,M^i)=(\mu_{Y_iX}^{-1/2} Y_iXT,M^i)\simeq (\mu_{Y_iX}^{-1/2}Z_iY_iXT,M)
\]
and we have
\begin{align*}
 X^*Y_i^*Z_i^*M Z_iY_iX&=  X^*Y_i^*M^i Y_iX
\le \sum_i  X^*Y_i^*M^i Y_iX\\
&=  X^*V^*N VX
=  X^*M X\le M.
\end{align*}
It follows that $Z_iY_iX\in \mathcal L_M$ and $(\mu_iX_iS,M^i)\in S_{\Fe_{T,M}}$.

Conversely, let $\mathcal I$ be the set of all $i$ such that $\mu_i=\Tr S^*X_i^*X_iS\ne 0$ and
suppose that  $(\mu_i^{-1/2}X_iS,M^i)\simeq F^i\in \Fe_{T,M}$ for all $i\in \mathcal I$. Then for $i\in \mathcal I$, $F^i\simeq (\mu_{Y_i}^{-1/2} Y_iT,M)$  for some $Y_i\in \mathcal L_M$, so that 
\[
S^*X_i^*M^iX_iS=s_iT^*Y_i^*MY_iT,
\]
where $s_i=\mu_i/\mu_{Y_i}$. Let $M=J^*EJ$ be a minimal Naimark representation of $M$. By Lemma \ref{lemma:L_M}, there are some $C_i\in \{E\}'$, $\|C_i\|\le 1$, such that $C_iJ=JY_i$, $i\in \mathcal I$. Note that then $J^*C_iJ=J^*JY_i=Y_i$ and $J^*C_i^*C_iJ=Y_i^*J^*JY_i= Y_i^*Y_i$. We then have for all $j=1,\dots,n$,
\begin{align*}
S^*N_jS&=\sum_{i\in\mathcal I} S^*X_i^*M_j^iX_iS=\sum_{i\in \mathcal I} s_i
T^*Y_i^*M_jY_iT=\sum_{i\in \mathcal I}s_i T^*Y_i^*J^*E_jJY_iT\\
&=\sum_{i\in \mathcal I} s_iT^*J^*C_i^*C_iE_jJT=T^*J^*CE_jJT,
\end{align*}
where $C:=\sum_{i\in \mathcal I} s_i C_i^*C_i$. Put $t:=\sum_is_i$, then
$0\le C\le tI$ and $C\in \{E\}'$, hence $D_j:=t^{-1}J^*CE_jJ$ satisfies 
$0\le D_j\le M_j$ and 
\[
\sum_j D_j=t^{-1}J^*CJ=t^{-1}\sum_{i\in \mathcal I}s_i Y_i^*Y_i\in B(\Ha_0).
\]
By Lemma \ref{lemma:char_extreme}, there is some $Y\in B(\Ha_0)$ such that 
$D_j=Y^*M_jY$ and it is clear that then $Y\in \mathcal L_M$. We therefore have
\[
S^*N_jS=T^*J^*CE_jJT=tT^*D_jT=tT^*Y^*M_jYT,\quad j=1,\dots,n,
\]
so that $(S,N)\simeq (t^{1/2}YT,M)\in S_{\Fe_{T,M}}$.
This proves (iii). 

\end{proof}

\begin{ex} Let $M\in \Me$ be a projection valued, then it is $B(\Ha_0)$-extremal by 
\ref{coro:PVM}. Let $T\in \Te$ be surjective. It is easy to see that $X^*MX\le M$ if and only if $X\in \{M\}'$, so that the face generated by $F$ is $\{(\mu_X^{-1/2} XT,M),\ X\in B(\Ha_0)\cap \{M\}'\}$.
\end{ex}

We also obtain  the following characterization of $B(\Ha_0)$-extremal POVMs.

\begin{coro}\label{coro:extrPOVM} An element $M\in \Me$ is $B(\Ha_0)$-extremal if and only if for any surjective $T\in \Te$ and any convex decomposition 
$F=\sum_i \lambda_i F^i$ of the corresponding process POVM $F=T^*MT$, we must have 
$F^i\simeq (\Ha_0,S_i,M)$ for some $S_i\in \Te$.  

\end{coro}

\begin{proof} Assume that $M$ is $B(\Ha_0)$-extremal and let $F=\sum_i\lambda_i F^i$. Any $F^i$ is contained in the face of $\Fe(\Ha,\Ka,n)$ generated by $F$, so that $F^i\simeq (\Ha_0,T_i,M)$ for some $T_i\in \Te$, by Proposition \ref{prop:smallestface}.
For the converse, let $M=\sum_i X_i^*M^iX_i$ be a proper $B(\Ha_0)$-convex combination and assume that $M$ fulfills the condition. 
Let $T\in \Te$ be surjective and let $F=T^*MT$, then $F=\sum_{i} \mu_i F^i$, with $F^i\simeq (\Ha_0, T_i,M^i)$ as in  (\ref{eq:muTiFi}). Note that $\mu_i>0$ for all $i$ since all $X_i$ are invertible. On the other hand,  $F^i\simeq (\Ha_0,S^i,M)$ by assumption. 
By Lemma \ref{lemma:minimal_testers}, $M\sim_{B(\Ha_0)} M^i$  for all $i$. It follows that $M$ is $B(\Ha_0)$-extremal.

\end{proof}

\section*{Acknowledgement} This work was supported by the grants VEGA 2/0125/13 and 2/0059/12, as well as the Research and Development Support Agency under the
contract No. APVV-0178-11.

\section*{Appendix:  $\Ae$-extremal and $\Ae$-pure POVMs}

We now give the proofs of Lemmas \ref{lemma:simored_face} - \ref{lemma:char_pure}. We mostly follow the arguments used in the C*-convex case. The first easy lemma and its corollary are proved similarly as e.g. in \cite{lopa1981some}.

\begin{lemma}\label{lemma:nto2} Let $M\in \Me(\Ha,n)$, 
$M=\sum_i X_i^*M^iX_i$ be an $\Ae$-convex combination of $M^1,\dots,M^m\in \Me(\Ha,n)$. Then there is some $N\in \Me(\Ha,n)$ and $Y_1\in \Ae$ such that 
$M=X_1^*M^1X_1+Y_1^*NY_1$ is an $\Ae$-convex combination of $M^1$ and $N$, proper if all $X_i$ are invertible.

\end{lemma}

%
%

\begin{coro} $M\in \Me(\Ha,n)$ is $\Ae$-extremal if and only if whenever 
$M=XM^1X+YM^2Y$ for some positive invertible $X,Y\in \Ae$, $X^2+Y^2=I$ and $M^1,M^2\in \Me(\Ha,n)$, we must have $M^i\sim_\Ae M $.

\end{coro}

\noindent{\it Proof of Lemma \ref{lemma:simored_face}.} Here we use similar techniques as in \cite{famo1993Cstar}. We first prove the assertion for $M=X_1N^1X_1+X_2N^2X_2$, with $X_1,X_2\ge 0$. If both $X_1$ and $X_2$ are invertible, $N^1,N^2\in \Fe$ by the definition of an $\Ae$-face. So suppose that, say, $X_1$ is not invertible.
For any  $\lambda\in (0,1)$, 
\[
M=X_1N^1X_1+ (\lambda X_2)N^2(\lambda X_2)+(\sqrt{1-\lambda^2} X_2)N^2(\sqrt{1-\lambda^2} X_2).
\]
Put $Y_\lambda=X_1^2 + (1-\lambda^2)X_2^2=I-\lambda^2X_2^2$ and 
\[
N^\lambda:=Y_\lambda^{-1/2}[X_1N^1X_1+(\sqrt{1-\lambda^2} X_2)N^2(\sqrt{1-\lambda^2} X_2)]Y_\lambda^{-1/2},
\]
note that  since $X_2^2\le I$, $Y_\lambda$ is invertible. We then have $N^\lambda\in \Me(\Ha,n)$ and 
\[
M=(\lambda X_2)N^2(\lambda X_2)+ Y_\lambda^{1/2}N^\lambda Y_\lambda^{1/2}.
\]
If $X_2$ is invertible, this is a proper $\Ae$-convex combination, so that  
$N^2,N^\lambda\in \Fe$. If $X_2$ is not invertible, we repeat the same construction, with $Y_\lambda^{1/2}N^\lambda Y_\lambda^{1/2}$ as the second element. Since $Y_\lambda$ is invertible, we obtain $M$ as a proper convex combination of $N^\lambda$ and some element $K^\lambda\in \Me(\Ha,n)$, so that again $N^\lambda\in \Fe$, and this holds for any $\lambda\in (0,1)$.
Since $\Fe$ is compact, there is some sequence $\lambda_n\to 1$, such that 
$N^{\lambda_n}$ converges to some element $N\in \Fe$.  
It follows that 
\[
\lim_n N^{\lambda_n}=Q_1N^1Q_1+Q_1^\perp N^2Q_1^\perp\in \Fe.
\]

Let now $A=\sum_{j=1}^m X_j^*N^jX_j$ and choose any $j\in \{1,\dots,m\}$. Then by Lemma \ref{lemma:nto2}, there is some $K^j\in \Me(\Ha,n)$ and $Y_j\in \Ae$
 such that $M=X_j^*N^jX_j+Y_j^*K^jY_j$ is an $\Ae$-convex combination. Let $X_j=U_j|X_j|$, $Y_j=V_j|Y_j|$ be polar decompositions, with $U_j$, $V_j$ unitary elements in $\Ae$. Then 
$M=|X_j|M^j|X_j|+|Y_j|R^j|Y_j|$, with $M^j=U_j^*N^jU_j$, $R^j=V_j^*K^jV_j\in \Me(\Ha,n)$. By the first part of the proof, for $P_j$ the support of $|X_j|$, 
\[
P_jM^jP_j+   P_j^\perp R^jP_j^\perp=P_jU^*_jN^jU_jP_j+
P^\perp R^j P_j^\perp\in \Fe
\]
and hence also $Q_jN^jQ_j+Q_j^\perp L^j Q_j^\perp\in \Fe$, where $L^j=U_j^*R^jU_j\in \Me(\Ha,n)$ and $Q_j=U_jP_jU^*_j$ is the range projection of $X_j$.

\hfill $\qed$

\noindent{\it Proof of Lemma \ref{lemma:char_extreme}.}
Let $M$ be $\Ae$-extremal and let $0\le D_i\le M_i$, $\sum_iD_i=:D\in \Ae$.
Let $Q$ be the support projection of $D$ and let 
\[
N_i=D^{-1/2}D_iD^{-1/2}+\tfrac 1n (I-Q), \quad i=1,\dots,n. 
\]
Then $N=\{N_1,\dots,N_n\}$
 is an element in $\Me(\Ha,n)$ and  $D_i=D^{1/2}N_iD^{1/2}$. Similarly, let $D'_i=M_i-D_i$, then $0\le D_i'\le M_i$, $\sum_i D_i'=I-D\in \Ae^+$ and there is some $N'\in \Me(\Ha,n)$ such that $D_i'=(I-D)^{1/2}N_i'(I-D)^{1/2}$. It follows that 
\begin{equation}\label{eq:char_proof}
M_i=D_i+D_i'=D^{1/2}N_iD^{1/2}+(I-D)^{1/2}N_i'(I-D)^{1/2},
\end{equation}
so that $M$ is an $\Ae$-convex combination of $N$ and $N'$. Using Lemmas \ref{lemma:extreme_face} and \ref{lemma:simored_face}, we obtain $M\sim_\Ae QNQ+Q^\perp N'Q^\perp$, so that there is some unitary $U\in \Ae$ such that 
\[
D_i=D^{1/2}(QN_iQ+Q^\perp N_i'Q^\perp)D^{1/2}=D^{1/2}U^*M_iUD^{1/2}=X^*M_iX,
\]
with $X=UD^{1/2}\in \Ae$. 

Conversely, assume that the condition holds and let $M=\sum_j Y_j^*M^jY_j$ be a proper $\Ae$-convex combination. Then for all  $i,j$, $0\le Y_j^*M_i^jY_j\le M_i$, hence there is some $X_j\in \Ae$ such that 
$Y_j^*M_i^jY_j=X_j^*M_iX_j$. Summing over $i$, we obtain $Y_j^*Y_j=X_j^*X_j$, so that $Y_j=U_jX_j$ for some unitary $U_j\in \Ae$ 
 and since $Y_j$ is invertible, so is $X_j$. Hence $U^*_jM^jU_j=M$, for all $j$, and $M$ is $\Ae$-extremal.

\hfill $\qed$

\begin{lemma}\label{lemma:fixedp} Let $\Phi:B(\Ha)\to B(\Ha)$ be a unital completely positive map. Let $\Fe_\Phi:=\{X\in B(\Ha), \Phi(X)=X\}$ and let $\Fe_\Phi'$ be the commutant of $\Fe_\Phi$.  If $\Fe_\Phi'=\mathbb CI$, then $\Phi=id_\Ha$ is the identity map on $\Ha$.

\end{lemma}

\begin{proof}  Clearly,  it is enough to prove that $\Fe_\Phi$ is a subalgebra, equivalently, $X^*X\in \Fe_\Phi$ for any $X\in \Fe_\Phi$. In proving this, we follow
 \cite[Remark 2]{arveson1972II}.

Let $\Psi$ be the pointwise limit of $n^{-1}\sum_{k=0}^{n-1}\Phi^k$, where $\Phi^k=\Phi\circ\dots\circ\Phi$ is the $k$-fold composition. Then $\Psi$ is an idempotent unital completely positive map and $\Psi\circ\Phi=\Phi\circ\Psi=\Psi$, $\Fe_\Phi\subseteq \Fe_\Psi$. 
Let $Q$ be the support projection of $\Psi$, then by \cite[Lemma 1]{arveson1972II}, $Q\in \Fe_\Psi'\subseteq \Fe_\Phi'$, so that $Q=I$ and
 $\Psi$ is faithful.
 If $X\in \Fe_\Phi$, then by Schwarz inequality, $X^*X=\Phi(X)^*\Phi(X)\le \Phi(X^*X)$. Since $\Phi(X^*X)-X^*X\ge 0$ and 
 $\Psi(\Phi(X^*X)-X^*X)=\Psi(X^*X)-\Psi(X^*X)=0$, we have $X^*X\in \Fe_\Phi$.

\end{proof}

\noindent{\it Proof of Lemma \ref{lemma:char_pure}.}
Suppose $M$ is  $\Ae$-pure and let $D_i\in B(P_i\Ha)$ be such that $D:=\sum_i D_i\in \Ae$. 
We may suppose that $D_i$ is self-adjoint, by replacing  $D_i$ by $\tfrac12(D_i+D_i^*)$ if necessary.
Moreover, there are some $t,s\ge 0$ such that $0\le tM_i+sD_i\le M_i$ for all $i$ and since  $\sum_i(tM_i+sD_i)=tI+sD\in \Ae$, we may suppose $0\le D_i\le M_i$. 

By Lemma \ref{lemma:char_extreme}, there are some $X,Y\in \Ae$ such that $D_i=X^*M_iX$ and 
$M_i-D_i=Y^*M_iY$ for all $i$. Summing over $i$, we obtain $X^*X+Y^*Y=I$. Let $\Phi(A)=X^*AX+Y^*AY$ for $A\in B(\Ha)$, then $\Phi$ is a unital completely positive map and $M_i\in \Fe_\Phi$ for all $i$. Clearly, $\Ae'\subseteq \Fe_\Phi$, so that 
$\Fe_\Phi'\subseteq \Ae''=\Ae$. Let $Q\in \Fe_\Phi'$ be a projection, then $Q$ must commute with all $M_i$. Since $Q\in \Ae$ and $M$ is $\Ae$-pure, this implies $Q=0$ or $I$. Hence $\Fe'_\Phi=\mathbb CI$ and $\Phi=id_\Ha$, by Lemma \ref{lemma:fixedp}.
Since $X$ and $Y$ are Kraus operators of $\Phi$, we must have $X=z I$ for some $z\in \mathbb C$, $|z|\le 1$. Thus $D_i=X^*M_iX=|z|^2M_i$, $i=1,\dots, n$.

Conversely, suppose that the condition holds, then $M$ is $\Ae$-extremal by Lemma \ref{lemma:char_extreme}. If $0\ne P\in \Ae$ is a projection commuting with all $M_i$, then $PM_i\in B(P_i\Ha)$ and the condition implies that $PM_i=\lambda M_i$ for all $i$
and some $\lambda\in [0,1]$. Hence $P=I$ and $M$ is $\Ae$-pure.

\hfill $\qed$


\end{document}